\documentclass{article}
\usepackage{spconf,amsmath,graphicx,hyperref}
\usepackage{amsmath,amssymb,amsfonts}
\usepackage{algorithmic}
\usepackage{graphicx}
\usepackage{textcomp}
\usepackage{xcolor}
\usepackage{times}
\usepackage{enumerate}
\usepackage{amsmath,amssymb,amsfonts}
\usepackage{mathtools} 

\usepackage{cleveref}
\usepackage{mathrsfs}
\usepackage{comment}
\usepackage{cite}

\def\BibTeX{{\rm B\kern-.05em{\sc i\kern-.025em b}\kern-.08em
    T\kern-.1667em\lower.7ex\hbox{E}\kern-.125emX}}

\usepackage{placeins}
\usepackage{amsthm}
\newtheorem{definition}{Definition}

\newtheorem{theorem}{Theorem}
\newtheorem{corollary}{Corollary}
\theoremstyle{definition}
\newtheorem{remark}{Remark}

\title{Risk level dependent Minimax Quantile lower bounds for Interactive Statistical Decision Making}
\name{Raghav Bongole, Amirreza Zamani, Tobias J. Oechtering, and Mikael Skoglund
\thanks{This work is supported by the Knut and Wallenberg Foundation.}}
\address{Division of Information Science and Engineering (ISE)\\
KTH Royal Institute of Technology\\
\texttt{\{bongole, amizam, oech, skoglund\}@kth.se}}
\begin{document}
%
\maketitle
\begin{abstract}
Minimax risk and regret focus on expectation, missing rare failures critical in safety-critical bandits and reinforcement learning. Minimax quantiles capture these tails. Three strands of prior work motivate this study: minimax-quantile bounds restricted to non-interactive estimation; unified interactive analyses that focus on expected risk rather than risk level specific quantile bounds; and high-probability bandit bounds that still lack a quantile-specific toolkit for general interactive protocols. To close this gap, within the interactive statistical decision making framework, we develop high-probability Fano and Le Cam tools and derive risk level explicit minimax-quantile bounds, including a quantile-to-expectation conversion and a tight link between strict and lower minimax quantiles. Instantiating these results for the two-armed Gaussian bandit immediately recovers optimal-rate bounds.
\end{abstract}
\section{Introduction}
\label{sec:intro}

The concept of minimax risk has become a staple in learning theory and statistics \cite{ma2024high,audibert2009minimax,azar2017minimax}. In interactive or online learning settings, the minimax risk is often replaced by its counterpart, the minimax regret. Both capture what the best algorithm can achieve when faced with the worst-case environment. In interactive decision-making problems such as bandits and reinforcement learning (RL), there has been considerable work either designing algorithms that achieve near minimax-optimal performance or proving, mainly via information-theoretic principles, lower bounds on minimax criteria \cite{jin2018q, lattimore2020bandit, foster2021statistical}. Such lower bounds are converse results: no algorithm can perform better than the stated limit. The minimax risk widely used in these works focuses on expected loss, where the expectation is taken over all trajectories of the interaction. Although well-studied with many tools such as Le Cam’s and Fano’s methods \cite{lecam1973convergence,polyanskiy2014lecture,cover1999elements,yu1997assouad,chen2016bayes,duchi2013distance,chen2024assouad} and the DEC framework \cite{foster2021statistical}, this criterion can miss information crucial in safety-critical applications: it is agnostic to tail behavior. For instance, two algorithms can share the same expected loss yet have very different $95\%$ quantiles, leading to markedly different risk profiles at high confidence. To study tails in a minimax sense, minimax quantiles provide a natural tool. Ma et al. \cite{ma2024high} develop high-probability techniques for minimax quantile lower bounds in estimation (non-interactive) problems, adapting classical methods to yield risk level $\delta$-explicit tail guarantees. However, their analysis is limited to settings without interaction.
On the interactive front, Chen et al. \cite{chen2024assouad} introduce the interactive statistical decision making (ISDM) framework to unify passive estimation and interactive protocols, and to systematize information-theoretic lower bounds (e.g. Fano and Le Cam) in interactive environments. Their focus is primarily on minimax expected risk; a $\delta$-explicit, quantile-focused treatment within interactive decision making has remained comparatively underexplored. Prior bandit literature contains high-probability lower bounds (e.g., \cite{lattimore2020bandit, gerchinovitz2016refined}), but these results do not by themselves provide a unifying route to minimax quantile lower bounds in general interactive models. In this work, we explore techniques that lower bound the minimax quantile in interactive problems, thereby illuminating fundamental limits on tail behavior under interaction. Inspired by Ma et al. \cite{ma2024high} for non-interactive minimax quantiles and by Chen et al. \cite{chen2024assouad} for interactive information-theoretic lower bounds, we derive counterparts that operate directly in ISDM and make the dependence on the risk level $\delta$ explicit. We also develop user-friendly tools that show how the risk level affects tails. 

Specifically, we make the following contributions:
i. We introduce a risk level  $\delta$-explicit minimax–quantile lens for ISDM that unifies non-interactive (estimation) and interactive analyses, enabling direct comparison and transfer of high-probability lower bounds.
ii. We show that minimax and lower minimax quantiles coincide for essentially all confidence levels, so one can work with the lower minimax quantile and then lift guarantees to minimax quantiles.
iii. We give a quantile to expectation conversion in ISDM that turns any $\delta$-level high-probability lower bound into an in-expectation minimax lower bound.
iv. We develop unified tools: an  interactive Fano method and a high-probability interactive Le Cam method, yielding $\delta$-explicit lower minimax quantiles generalizing metric packings.
v. We demonstrate ready applicability: a direct instantiation for the two-armed Gaussian bandit yields 
$\delta$-explicit bounds recovering the $\sqrt{T\log(1/\delta)}$ scaling.

\section{Framework}
We use the framework from \cite{chen2024assouad} as it includes statistical estimation problems as well as interactive decision making problems. Suppose that the space of all outcomes is denoted by $\mathcal{X}$ and $\mathcal{M}$ comprises the set of all models. The interaction process proceeds as follows: A decision maker chooses a decision/algorithm $\textup{ALG}$ from a set $\mathcal{D}$, while the environment chooses a model $M \in \mathcal{M}$. Based on $M$ and $\textup{ALG}$, there is an induced probability distribution $\mathbb{P}^{M,\textup{ALG}}$. An outcome $X$ is generated from the induced distribution: $X \sim \mathbb{P}^{M,\textup{ALG}}$. The goal of the decision maker is to minimize the risk specified by non-negative function $L(M,X)$ by choosing an appropriate $\textup{ALG}$. Thus, an ISDM problem comprises a 4-tuple $(\mathcal{X}, \mathcal{M}, \mathcal{D}, \mathcal{L})$.

A natural criterion is to evaluate an algorithm in the worst case over models, formalized via the minimax risk.


\begin{definition}[Minimax Risk]
The minimax risk \(\mathfrak{M}\) captures the least achievable risk against the most challenging model in \(\mathcal{M}\):
$
\mathfrak{M}
\;:=\;
\inf_{\mathrm{ALG}\in\mathcal{D}}\;
\sup_{M\in\mathcal{M}}\;
\mathbb{E}^{M,\mathrm{ALG}}\!\bigl[L(M,X)\bigr].
$
\end{definition}

Motivated by the observation that expectation–based minimax risk \(\mathfrak{M}\) may mask rare, large losses \cite{ma2024high}, we quantify tail performance via the \((1-\delta)\)–quantile of the loss, which we define below.

\begin{definition}[Quantile]
For a given level \(\delta\in(0,1]\), the \((1-\delta)\)-th quantile is defined as
\begin{align*} \textup{Quantile}&(1-\delta, \mathbb{P}^{M,\textup{ALG}}, L) = \\ &\inf \left\{r \in [0, \infty]: \mathbb{P}^{M,\textup{ALG}}(L(M,X)>r) \leq \delta\right\}. \end{align*}
\end{definition}

The minimax quantile is the smallest \((1-\delta)\)–quantile that any algorithm can guarantee uniformly over the model class.

\begin{definition}[Minimax Quantile]
For a given level \(\delta\in(0,1]\), the (strict) minimax quantile is
\begin{align*}
\mathfrak{M}(\delta)
\;:=\;
\inf_{\mathrm{ALG}\in\mathcal{D}}\;
\sup_{M\in\mathcal{M}}\;
\textup{Quantile}&(1-\delta, \mathbb{P}^{M,\textup{ALG}}, L) .
\end{align*}
\end{definition}

Following~\cite{ma2024high}, we introduce the lower minimax quantile \(\mathfrak{M}_{-}(\delta)\), a relaxation of the minimax quantile that is more amenable to analysis: it admits \(\delta\)-explicit lower bounds via high-probability analogues of Le Cam’s and Fano’s methods under simple separation conditions.

\begin{definition}[Lower Minimax Quantile]
For a given level \(\delta\in(0,1]\), the lower minimax quantile is given by $\mathfrak{M}\_(\delta) =$
\begin{align*} \inf \bigg\{r \in [0, \infty]: \inf_{\textup{ALG} \in \mathcal{D}} \sup_{M \in \mathcal{M}} \mathbb{P}^{M,\textup{ALG}}(L(M,X)> r) \leq \delta \bigg\}. \end{align*}
\end{definition}

\newcommand{\ALG}{\mathrm{ALG}}
\newcommand{\TV}{\operatorname{TV}}
\newcommand{\KL}{\operatorname{KL}}

\section{Main Results}

\noindent\textbf{Conventions.}
All probabilities/expectations are under the interactive law $\mathbb{P}^{M,\ALG}$ and $\mathbb{E}^{M,\ALG}$. We abbreviate total variation distance, Kullback–Leibler divergence and mutual information by $\TV$, $\KL$ and MI.

\subsection{Bridges between tails, quantiles, and expectation.}
In the following result, we link minimax quantiles to the minimax risk in ISDM.

\begin{theorem}[Quantile-to-expectation in ISDM]
Let $\mathfrak{M}(\delta)$ be the minimax quantile and let $\mathfrak{M}$ be the minimax risk; then for every $\delta\in(0,1]$, we have
\[
\mathfrak{M}
=\inf_{\ALG\in\mathcal{D}}\ \sup_{M\in\mathcal{M}}\
\mathbb{E}^{M,\ALG}\!\bigl[L(M,X)\bigr]\ \ge\ \delta\, \mathfrak{M}(\delta).
\]
This shows that a $(1-\delta)$–quantile lower bound immediately implies a $\delta$–scaled expectation lower bound.
\end{theorem}

\begin{proof}[Proof]
The ISDM proof extends verbatim from the estimation setting \cite[Prop.~2]{ma2024high}: apply $\mathbb{E}[L]\ge r\mathbb{P}(L>r),$ for $r=\textup{Quantile}(1-\delta, \mathbb{P}^{M,\textup{ALG}}, L)$ for each fixed $(M,\ALG)$ and pass $\sup_M$ and $\inf_{\ALG}$ through.
\end{proof}

We next relate the minimax and lower minimax quantiles in ISDM and provide a tail-to-quantile conversion.

\begin{theorem}[Lower minimax quantile relation in ISDM]
\label{thm:lower-minimax-quantile-relation}
In the ISDM setting (with all probabilities taken under $\mathbb{P}^{M,\ALG}$), for every $\delta\in(0,1]$ and every $\xi\in(0,\delta)$, we have
\[
\mathfrak{M}_{-}(\delta)\ \le\ \mathfrak{M}(\delta)\ \le\ \mathfrak{M}_{-}(\delta-\xi).
\]
Consequently, $\mathfrak{M}(\delta)=\mathfrak{M}_{-}(\delta)$ for all $\delta\in(0,1]$ except a countable set. Moreover (tail-to-quantile conversion): if for some $r\ge 0$,
$
\inf_{\ALG}\ \sup_M\ \mathbb{P}^{M,\ALG}\!\big(L(M,X)>r\big)\;>\;\delta,$ then
\[
\quad
\mathfrak{M}(\delta)\ \geq\ \mathfrak{M}_{-}(\delta)\ \geq r.
\]
Thus, lower and strict minimax quantiles essentially coincide, and any minimax tail lower bound yields a minimax quantile lower bound.
\end{theorem}

\begin{proof}
We extend the relation to the ISDM setting. The proof proceeds as in \cite[Th.~4]{ma2024high} with all probabilities taken under $\mathbb{P}^{M,\ALG}$: the inequalities $\mathfrak{M}_{-}(\delta)\le \mathfrak{M}(\delta)\le \mathfrak{M}_{-}(\delta-\xi)$ follow from the same quantile–tail and $\xi$–slack steps, and “equality a.e.” follows by the monotonicity of $\mathfrak{M}_{-}$; none of these arguments use non-interactivity, so each step extends verbatim. The tail-to-quantile claim follows since $r\mapsto \inf_{\ALG}\sup_M \mathbb{P}^{M,\ALG}(L>r)$ is non-increasing.
\end{proof}

\subsection{Tools based on the Fano method.}
We next derive a high-probability interactive Fano method that yields lower bounds on the minimax quantile under an f-divergence condition.

\begin{theorem}[High-probability interactive Fano $\Rightarrow$ lower minimax quantile]
\label{High-probability interactive Fano}
Fix an $f$-divergence $D_f$, a prior $\mu\in\Delta(\mathcal{M})$, and a risk level $\Delta>0$. For any $Q\in\Delta(\mathcal{X})$, set
\[
\begin{aligned}
&\bar\rho_{\Delta,Q}
:= \mathbb{P}_{M\sim\mu,\,X\sim Q}\!\bigl(L(M,X)\le \Delta\bigr),\\[-0.25ex]
&d_{f,\epsilon}(p)
:=
\begin{cases}
D_f\!\bigl(\mathrm{Bern}(1-\epsilon),\,\mathrm{Bern}(p)\bigr), & p\le 1-\epsilon,\\
0, & p>1-\epsilon,
\end{cases}
\\[-0.25ex]
&
\epsilon^\star
:=\!\!\!\!\!\sup_{\substack{Q\in\Delta(\mathcal{X})\\ \epsilon\in[0,1]}}\!\!
\Bigl\{\!\epsilon\!:\!\!
\sup_{\ALG}\ \mathbb{E}_{M\sim\mu}\!\bigl[D_f(\mathbb{P}_{M,\ALG}\|Q)\bigr]\!
<\! d_{f,\epsilon}\!\bigl(\bar\rho_{\Delta,Q}\bigr)\!\Bigr\}.
\end{aligned}
\]
Then, for all $\delta\in[0,\epsilon^\star)$, we have $\ \mathfrak{M}_{-}(\delta)\ \ge\ \Delta$.
\end{theorem}

\begin{proof}[Proof sketch]
Apply the interactive Fano method \cite[Thm.~2]{chen2024assouad} with the non-strict
success event $\mathbf{1}\{L\le \Delta\}$ in place of $\mathbf{1}\{L<\Delta\}$.
With this change, proceeding as in the proof of \cite[Thm.~2]{chen2024assouad}, for every $\ALG$,
$
\mathbb{P}_{M\sim\mu,\,X\sim \mathbb{P}_{M,\ALG}}\!\bigl(L(M,X)>\Delta\bigr)\ \ge\ \epsilon
\quad\text{whenever}\quad
\mathbb{E}_{M\sim\mu} D_f(\mathbb{P}_{M,\ALG}\|Q) < d_{f,\epsilon}(\bar\rho_{\Delta,Q}).
$
By the definition of $\epsilon^\star$, the same $(Q,\epsilon)$ works uniformly over all algorithms, hence $\inf_{\ALG}\sup_M \mathbb{P}(L(M,X)>\Delta)\ge \epsilon$ for every $\epsilon<\epsilon^\star$. Using Theorem \ref{thm:lower-minimax-quantile-relation} (tail-to-quantile conversion) completes the proof.
\end{proof}

\begin{remark}
We emphasize that Theorem~\ref{High-probability interactive Fano} strengthens the interactive Fano guarantee \cite[Thm.~2]{chen2024assouad} in two ways. 
First, it is uniform in the algorithm: the threshold $\epsilon^\star$ is defined with a $\sup_{\ALG}$ inside the divergence condition, so the lower bound $\mathfrak{M}_{-}(\delta)\ge \Delta$ holds simultaneously for every $\delta<\epsilon^\star$, yielding an entire curve in $\delta$. 
Second, it aligns the event with the target quantile by using the non-strict success set $\{L\le \Delta\}$ giving a bound on the strict-tail $\mathbb{P}(L>\Delta)$, a stronger result. 
In contrast, \cite[Thm.~2]{chen2024assouad} provides a per–algorithm bound written via $\rho_{\Delta,Q}=\mathbb{P}(L<\Delta)$ giving a bound on the weak-tail $\mathbb{P}(L\geq\Delta)$ and typically requires re-optimizing the reference $Q$ and level for each desired $\delta$, without a uniform-in-$\ALG$ threshold.
\end{remark}

The following corollary is a specialization of Theorem~\ref{High-probability interactive Fano}, yielding a mutual information (MI) based criterion.
\begin{corollary}[MI criterion $\Rightarrow$ lower minimax quantile]
\label{cor:info-criterion}
Let $\mu\in\Delta(\mathcal{M})$ and $\Delta>0$. Define
\[
p_{\max}\ :=\ \sup_{x\in\mathcal{X}}\ \mu\!\bigl(\{M:\ L(M,x)\le \Delta\}\bigr)\ <\ 1.
\]
For each algorithm $\ALG$, let $I_{\mu,\ALG}(M;X)$ denote the mutual information under $M\sim\mu$ and $X\sim \mathbb{P}^{M,\ALG}$. If there exists $\epsilon\in(0,1]$ such that, for all $\ALG$,
\[
1+\frac{I_{\mu,\ALG}(M;X)+\log 2}{\log p_{\max}}\ \ge\ \epsilon,
\]
then for all $\delta\in(0,\epsilon)$ we have $\ \mathfrak{M}_{-}(\delta)\ \ge\ \Delta$.
\end{corollary}

\begin{proof}[Proof sketch]
Apply Theorem~\ref{High-probability interactive Fano} with $f=\mathrm{KL}$ and $Q=\mathbb{E}_{M \sim \mu}\,\mathbb{P}^{M,\ALG}$; use $\bar\rho_{\Delta,Q}\le p_{\max}$ and verify the f-divergence condition by choosing $\epsilon' < \epsilon$ so that, for all $\ALG$, we have, $I_{\mu,\ALG}(M;X)+\log 2 \le (1-\epsilon)\log\!\bigl(1/p_{\max}\bigr)\le \ (1-\epsilon')\log\!\bigl(1/p_{\max}\bigr) < \KL(1-\epsilon'||\bar\rho_{\Delta,Q}) $, as in \cite[Prop.~3]{chen2024assouad}.
\end{proof}

\begin{remark}
Ma et al.’s high-probability Fano lemma \cite[Lem.~7]{ma2024high} is stated for estimation via a finite packing and a separation parameter $\eta$ linked through a problem-specific function $g$, producing bounds of the form $\mathfrak{M}_{-}(\delta)\ge g(\eta)$ for $\delta\in(0,\epsilon)$. 
Corollary~\ref{cor:info-criterion} reaches an analogous high-probability conclusion without constructing a packing or enforcing pairwise separation: it generalizes the metric/separation premise to the condition
$p_{\max}:=\sup_{x\in\mathcal{X}}\mu\!\bigl(\{M:\,L(M,x)\le \Delta\}\bigr)<1$ (cf.~\cite[Prop.~3]{chen2024assouad}),
while remaining uniform over interactive algorithms and valid for arbitrary risk levels $\Delta$.
Also, Corollary~\ref{cor:info-criterion} holds even in the \emph{interactive} setting, whereas \cite[Lem.~7]{ma2024high} is formulated for \emph{estimation} (non-interactive) problems.
\end{remark}

\subsection{Tools based on Le Cam's method.}
We now present high-probability Le Cam tools under interaction: closeness in TV distance or in KL divergence, with a uniform separation yields a minimax quantile lower bound.

\begin{theorem}[High-probability Le Cam under interaction (TV and KL)]
\label{thm:hp-lecam}
Let $\delta \in (0,\tfrac{1}{2})$ and fix $M_1,M_2 \in \Theta$. 
Assume the uniform separation $L(M_1,x)+L(M_2,x)\ge 2\Delta$ for all transcripts $x$. Then:
\begin{itemize}
\item[(a)] \emph{TV condition:} If $\TV(\mathbb{P}^{M_1,\ALG},\mathbb{P}^{M_2,\ALG}) < 1 - 2\delta$, then $\ \mathfrak{M}_{-}(\delta) \ge \Delta$.
\item[(b)] \emph{KL condition:} If $\KL(\mathbb{P}^{M_1,\ALG}\|\mathbb{P}^{M_2,\ALG})\!<\! \log\!\tfrac{1}{4\delta(1-\delta)}$ then $\ \mathfrak{M}_{-}(\delta) \ge \Delta$.
\end{itemize}
\end{theorem}

\begin{proof}
We extend the indicator/thresholding device of \cite[Lemma 5]{ma2024high} to the interactive ISDM setting by coupling it with the interactive Le Cam inequality.
Write $\mathbb{P}_i := \mathbb{P}^{M_i,\ALG}$ for the transcript law under $M_i$. The proof proceeds in four steps.
\emph{i) (indicator reduction).} Fix $a\in(0,1)$ and set $g_a(t):=\mathbf{1}\{t>a\Delta\}$. The separation implies $g_a(L(M_1,x))+g_a(L(M_2,x))\ge 1$ for all transcripts $x$.
\emph{ii) (two-point inequality under interaction).} By the interactive two-point Le Cam inequality for expected risk \cite[Prop.~4]{chen2024assouad}, applied to the bounded losses $g_a\!\circ\! L$ and the pair $(M_1,M_2)$,
$
\inf_{\ALG}\ \sup_{j\in\{1,2\}}
\mathbb{P}^{M_j,\ALG}\!\bigl(L(M_j,X)>a\Delta\bigr)
=
\inf_{\ALG}\ \sup_{j}\ \mathbb{E}^{M_j,\ALG}\!\bigl[g_a(L(M_j,X))\bigr]
\ \ge\ \tfrac{1}{2}\,\bigl(1-\\TV(\mathbb{P}_1,\mathbb{P}_2)\bigr).
$
\emph{iii) (TV condition $\Rightarrow$ tail and quantile).} Under the premise of (a), $\tfrac{1}{2}\,(1-\TV(\mathbb{P}_1,\mathbb{P}_2))>\delta$, hence,
$
\inf_{\ALG}\sup_{M\in\{M_1,M_2\}}\mathbb{P}^{M,\ALG}\!\bigl(L(M,X)>a\Delta\bigr)>\delta.
$
Using Theorem~\ref{thm:lower-minimax-quantile-relation} (tail-to-quantile conversion), $\mathfrak{M}_{-}(\delta)\ge a\Delta$; letting $a\uparrow 1$ yields $\mathfrak{M}_{-}(\delta)\ge \Delta$.
\emph{iv) (KL condition $\Rightarrow$ TV condition).} Proceeding  as in \cite[Cor.~6]{ma2024high}, if
$\KL(\mathbb{P}_1\|\mathbb{P}_2) < \log\!\tfrac{1}{4\delta(1-\delta)}$, the Bretagnolle–Huber inequality \cite{bretagnolle1979estimation} implies \\$\TV(\mathbb{P}_1,\mathbb{P}_2) < 1-2\delta$; thus, (a) applies proving (b). 
\end{proof}

\section{Relation to Prior Works}
\label{sec:related}

The Decision–Estimation Coefficient (DEC) \cite{foster2021statistical, foster2023model, chen2024assouad} is a problem-dependent complexity measure that packages classical lower-bound ideas into a function trading off loss separation and information indistinguishability. Intuitively, a large DEC certifies the existence of many models that are far in loss yet hard to distinguish under the data-collection protocol, forcing nontrivial tail risk/regret. In the ISDM framework, \cite{chen2024assouad} formalize two variants: the \emph{quantile-DEC} \(r\text{-}\mathrm{dec}^{q}\), tailored to weak-tail (\(\ge\)) guarantees at level \(\delta\), and the \emph{constrained DEC} \(r\text{-}\mathrm{dec}^{\mathrm{c}}\), a reference-model version that transfers hardness across classes. They prove high-probability lower bounds via \(r\text{-}\mathrm{dec}^{q}\) and relate \(r\text{-}\mathrm{dec}^{q}\) and \(r\text{-}\mathrm{dec}^{\mathrm{c}}\) as summarized below. Throughout we use up-to-constants notation: \(A \gtrsim B\) abbreviates \(A \ge c\,B - C\) for universal constants \(c>0\) and \(C\ge 0\) independent of \(T,\delta\), the algorithm, and the model class. All quantifiers remain exact.

For comparison with our strict quantiles (defined via \(P(L>r)\)), we introduce the weak quantile \(\mathrm{WQuantile}(1-\delta,\mathbb{P}^{M,\mathrm{ALG}},L):=\inf\{r\!\ge\!0:\mathbb{P}^{M,\mathrm{ALG}}(L(M,X)\ge r)\le \delta\}\), and the corresponding minimax weak quantile \(\mathfrak{M}_W(\delta):=\inf_{\mathrm{ALG}\in\mathcal{D}}\;\sup_{M\in\mathcal{M}}\;\mathrm{WQuantile}(1-\delta,\mathbb{P}^{M,\mathrm{ALG}},L)\). Weak-tail bounds (\(\ge\)) do not in general imply strict-tail bounds (\(>\)) without continuity (e.g., if \(P(L=r)>0\)). We use the standard mapping: if, for every \(\mathrm{ALG}\), there exists \(M\in\mathcal{M}\) with \(\mathbb{P}^{M,\mathrm{ALG}}(L\ge b)\ge p\), then \(\mathfrak{M}_W(p)\ge b\).
The high-probability lower bound of \cite[Thm.~D.3]{chen2024assouad} states that, for all \(\mathrm{ALG}\), there exists \(M\in\mathcal{M}\) such that \(\mathbb{P}^{M,\mathrm{ALG}}\!\bigl(L(M,X)\ge T\cdot r\text{-}\mathrm{dec}^{q}(\mathcal{M})\bigr)\gtrsim \delta\), whence the mapping yields \(\mathfrak{M}_W(\delta)\gtrsim T\cdot r\text{-}\mathrm{dec}^{q}(\mathcal{M})\). Moreover, \cite[Prop.~D.2]{chen2024assouad} relates the quantile-DEC to the constrained DEC via \(r\text{-}\mathrm{dec}^{q}(\mathcal{M})\gtrsim \max_{\widehat{M}\in\mathcal{M}^{+}} r\text{-}\mathrm{dec}^{\mathrm{c}}(\mathcal{M},\widehat{M})\), where \(\mathcal{M}^{+}\) denotes the reference models permitted in their constrained formulation. Combining these relations yields \(\mathfrak{M}_W(\delta)\gtrsim T\cdot \max_{\widehat{M}\in\mathcal{M}^{+}} r\text{-}\mathrm{dec}^{\mathrm{c}}(\mathcal{M},\widehat{M})\). Thus, the minimax weak quantile is controlled via the quantile-DEC and, through the relation above, the constrained DEC. As presented in \cite{chen2024assouad}, these results are not carried through to \(\delta\)-explicit statements for \(\mathfrak{M}_W(\delta)\) and do not yield strict-tail (\(>\)) minimax quantile bounds. In contrast, our results keep \(\delta\) explicit and certify strict minimax quantile lower bounds (see Cor.~\ref{cor:info-criterion}, Thm. \ref{thm:hp-lecam}).

\section{The two-armed Gaussian Bandit}
The following demonstrates the simplicity of our framework to derive high probability/quantile lower bounds for the minimax risk. 
In this section, we define a simple 2-arm bandit problem as an instance of ISDM and apply Theorem 5 to derive a lower bound on the lower minimax quantile. 

A two-armed Gaussian bandit can be viewed as an instance of the ISDM framework.
Let the model class be
$
\mathcal{M}
=\big\{\,\big(\mathcal{N}(\mu_1,1),\,\mathcal{N}(\mu_2,1)\big)\ :\ (\mu_1,\mu_2)\in\mathbb{R}^2\,\big\}.
$ Let the history space be $\mathcal{H}^t := (\{1,2\}\times\mathbb{R})^t$ and let
$\mathrm{ALG}$ be a sequence of stochastic decision kernels
$\{\Pi_t\}_{t=1}^T$ with $\Pi_t(\cdot \mid h^{t-1}) \in \Delta(\{1,2\})$ for each
$h^{t-1}\in\mathcal{H}^{t-1}$. Then, the interaction evolves as
$
\pi_t \sim \Pi_t(\cdot \mid H^{t-1}),\qquad
o_t \mid (\pi_t,H^{t-1}) \sim M_\mu(\pi_t)=\mathcal{N}(\mu_{\pi_t},1),
$
with $(o_t)$ conditionally independent of $H^{t-1}$ given $\pi_t$.
The (random) trajectory is
$
H^T := (\pi_1,o_1,\ldots,\pi_T,o_T)\in(\{1,2\}\times\mathbb{R})^T.
$
The data is then given by \(X=H^T\),
with distribution \(\mathbb{P}^{M,\mathrm{ALG}}\). Furthermore, let the loss be the regret given by \(L(M,X)  = \sum^T_{t=1} \mu^* - \mu_{\pi_t}\) where $\mu^* = \max(\mu_1,\mu_2)$.
\begin{theorem}
    For the two-armed Gaussian bandit problem above with $T\geq1$, for $\delta \in (0,1/2)$, we have $\mathfrak{M}_{-}(\delta) \geq \sqrt{\frac{T\log(1/4\delta(1-\delta))}{2}}.$
\end{theorem}

\begin{proof}
Fix any (possibly randomized, adaptive) algorithm. For a model \(M=(\mu_1,\mu_2)\) and trajectory \(x\) with actions \(A_{1:T}\), let \(\mathfrak{R}_T(M,x)\!=\!\sum_{t=1}^T(\mu^\star(M)-\mu_{A_t})\), where \(\mu^\star(M)=\max\{\mu_1,\mu_2\}\). Pick \(g>0\) and consider \(M_1:(\mu_1,\mu_2)\!=\!(+g/2,-g/2)\) and \(M_2\!:\!(\mu_1,\mu_2)\!=\!(-g/2,+g/2)\). If \(N_a(x)\) is the number of pulls of arm \(a\) in \(x\), then \(\mathfrak{R}_T(M_1,x)=g\,N_2(x)\) and \(\mathfrak{R}_T(M_2,x)=g\,N_1(x)\); thus, \(\mathfrak{R}_T(M_1,x)+\mathfrak{R}_T(M_2,x)=g(N_1(x)+N_2(x))=gT\), i.e., the separation condition holds with \(\Delta=gT/2\). Let \(\mathbb{P}_i\) be the trajectory law under \(M_i\) and set \(\Lambda_\delta=\log(1/(4\delta(1-\delta)))\). Using the divergence decomposition lemma \cite[Lemma~15.1]{lattimore2020bandit}, \(\mathrm{KL}(\mathbb{P}_1\|\mathbb{P}_2)=\sum_{a=1}^2 \mathbb{E}_{M_1}[N_a(T)]\,\mathrm{KL}(\mathcal N(\mu_a^{(1)},1)\Vert \mathcal N(\mu_a^{(2)},1))=(g^2/2)\,T.\)
Using Theorem~\ref{thm:hp-lecam}, if \(\mathrm{KL}(\mathbb{P}_1\|\mathbb{P}_2)<\Lambda_\delta\) and the separation holds, then \(\mathfrak{M}_{-}(\delta)\ge \Delta\). Hence, whenever \(g^2T/2<\Lambda_\delta\), we have \(\mathfrak{M}_{-}(\delta)\ge gT/2\). Choose any \(\eta\in(0,1)\) and take \(g=(1-\eta)\sqrt{2\Lambda_\delta/T}\), which yields \(\mathfrak{M}_{-}(\delta)\ge (1-\eta)\sqrt{(T/2)\Lambda_\delta}=(1-\eta)\sqrt{(T/2)\log(1/(4\delta(1-\delta)))}\); letting \(\eta\downarrow 0\) proves the claim.
\end{proof}

\begin{remark}
Chapter~17 of \cite{lattimore2020bandit} proves weak-tail bounds \(\mathbb{P}(L\ge r)\ge \delta\) with the same two-arm Gaussian scaling \(\sqrt{T\log(1/\delta)}\) (up to constants). In contrast, we obtain a \(\delta\)-explicit \emph{strong-tail} minimax-quantile lower bound via Theorem~\ref{thm:hp-lecam}; instantiating it for the two-arm Gaussian bandit yields our bound immediately. Moreover, the same theorem holds in ISDM and provides a template for analogous high-probability lower bounds beyond bandits.
\end{remark}

 \section{Conclusion and Future Work}
We introduce a $\delta$-explicit minimax-quantile lens for ISDM problems and supply interactive Fano and Le Cam tools that are straightforward to instantiate, recovering optimal scaling on two-armed bandits. It would be interesting to extend to the ISDM framework, \cite[Thm.~8]{ma2024high}, which transfers local minimax-risk lower bounds to the minimax quantile. Another future direction is to develop both \emph{$\delta$-explicit} weak minimax quantile lower bounds and strict minimax-quantile lower bounds in interactive settings using the quantile-DEC framework, for various problems including episodic RL.

\label{sec:refs}



\bibliographystyle{IEEEbib}
\bibliography{strings,refs}

\end{document}